\documentclass[12pt]{iopart}

\usepackage{soul}
\usepackage{url}
\usepackage{graphicx}
\expandafter\let\csname equation*\endcsname\relax
\expandafter\let\csname endequation*\endcsname\relax
\usepackage{booktabs}
\usepackage{algorithm}
\usepackage{algorithmic}
\usepackage{amsthm}
\usepackage{amssymb}
\usepackage{physics}
\usepackage{thmtools}
\usepackage{thm-restate}

\usepackage{framed}
\usepackage{bbm}

\usepackage{tikz}
\usetikzlibrary{positioning, shapes.geometric}

\usepackage{comment}
\usepackage[colorlinks,
            linkcolor=blue,
            anchorcolor=blue,
            citecolor=blue]{hyperref}
\usepackage{cleveref}

\newtheorem{mydef}{Definition}

\newtheorem{mythe}{Theorem}

\begin{document}

\title[Quantum Algorithm for Online Convex Optimization]{Quantum Algorithm for Online Convex Optimization}

\author[1]{Jianhao He$^1$, Feidiao Yang$^2$, Jialin Zhang$^2$, Lvzhou Li$^{1,*}$}

\address{$^1$ Institute of Quantum Computing and Computer Science Theory, School of  Computer Science and Engineering, Sun Yat-sen University, Guangdong Guangzhou, China}
\address{$^2$ Institute of Computing Technology, Chinese Academy of Sciences, Beijing, China}
\address{$^*$Author to whom any correspondence should be addressed.}
\ead{lilvzh@mail.sysu.edu.cn}
\vspace{10pt}
\begin{indented}
\item[]July 2021
\end{indented}

\begin{abstract}
We explore whether quantum advantages can be found  for the zeroth-order online convex optimization problem, which is also known as bandit convex optimization with multi-point feedback.
In this setting, given access to zeroth-order oracles (that is, the loss function is accessed as a black box that returns the function value for any queried input), a player attempts to minimize a sequence of adversarially generated convex loss functions.
This procedure can be described as a $T$ round iterative game between the player and the adversary.
In this paper, we present quantum algorithms for the problem and show for the first time that potential quantum advantages are possible for problems of online convex optimization.
Specifically, our contributions are as follows. (i) When the player is allowed to query zeroth-order oracles $O(1)$ times in each round as feedback, we give a quantum algorithm that achieves $O(\sqrt{T})$ regret without additional dependence of the dimension $n$, which outperforms the already known optimal classical algorithm only achieving $O(\sqrt{nT})$ regret.
Note that the regret of our quantum algorithm  has achieved the lower bound of classical first-order methods.  (ii) We show that for strongly convex loss functions, the quantum algorithm can achieve $O(\log T)$ regret with $O(1)$ queries as well, which means that the quantum algorithm can achieve the same regret bound as the  classical algorithms in the full information setting.
\end{abstract}

\vspace{2pc}

%
%
%
%
%

\section{Introduction}
\label{intro}

Convex optimization is a basic foundation for artificial intelligence, particularly for  machine learning.
While many ingenious algorithms have been developed for convex optimization problems \cite{Boyd2004,Bubeck15}, people still hunger for more efficient solutions in the era of big data.
Since quantum computing exhibits advantages over classical computing \cite{Shor1999,Grover1996,Harrow2009,Harrow2017}, people seek to employ  quantum computing techniques to accelerate the optimization process. 
On the one hand,  combinatorial optimization was shown to be  acceleratable  by using quantum techniques such as Grover's algorithm or quantum walks \cite{Grover1996,Ambainis2006,Durr2006,Durr1996,Mizel2009,Yoder2014,Sadowski2015,He2020}.
On the other hand, in the last few years, some significant quantum improvements were achieved for convex optimization in linear programming \cite{KP18,Li2019,Van2019}, second-order cone programming \cite{Kerenidis2019s,Kerenidis2019Svm,Kerenidis2019P}, quadratic programming \cite{kerenidis2020}, polynomial optimization \cite{Rebentrost2019}, and  semi-definite optimization \cite{KP18,Joran19,Fernando17,Fernando19,Joran17}.
Note that they are all special cases of convex optimization. In the last two years, quantum algorithms for general convex optimization were studied \cite{DeWolf18,Childs18}, where
the main  technique  used is the quantum gradient estimation algorithm \cite{Jordan05,Gilyen2019}.
Note that the studies mentioned above focus on improving offline optimization with quantum computing techniques. Recently, people began to consider applying quantum computing methods to online optimization problems. In 2020, two related results were given \cite{casale2020, wang2020}, where quantum algorithms for the best arm identification problem, a central problem in multi-armed bandit, were proposed. In the same year, the initial version of this paper was uploaded to arXiv \cite{he2020online}. The online problem considered in \cite{casale2020, wang2020} is discrete, while we study a continuous online problem, that is the online convex optimization problem. Besides the different settings, the basic quantum technique they used is the quantum amplitude amplification, while ours is the quantum phase estimation.

\subsection{Online convex optimization}\label{secOCO}

{Online convex optimization (OCO) }is an important framework in online learning, and particularly useful in sequential decision making problems, such as online routing, portfolio selection, and recommendation systems.
Online convex optimization is best understood as a $T$ round iterative game between a player and an adversary.
At every iteration $t$, the player generates a prediction $x_t$ from a fixed and known convex set $\mathcal{K}\subseteq\mathbb{R}^n$.
The adversary observes $x_t$ and chooses a convex loss function $f_t:\mathcal{K}\to \mathbb{R}$.
After that, the player suffers the loss $f_t(x_t)$.
Then, some information about the loss function $f_t$ is revealed to the player as feedback. The goal of the player is to minimize his \textit{regret}, which is defined as
\[
    \sum_{t=1}^T f_t(x_t)-\min_{x^*\in\mathcal{K}}\sum_{t=1}^T f_t(x^*).
\]

A good algorithm/strategy of the player should have a sublinear regret (that is, its regret is sublinear as a function of $T$) since this implies that as $T$ grows, the accumulated loss of the algorithm converges to that with the best fixed strategy in hindsight \cite{Hazan16,lattimore2020b}.

\begin{framed}
The variants of OCO mainly depend on the following four aspects:
\begin{enumerate}
    \item \textbf{The power of the adversary}
    \begin{enumerate}
        \item \textit{Completely adaptive adversary}: A completely adaptive adversary is allowed to choose loss functions $f_t$ after observing the player’s choice $x_t$.
        \item \textit{Adaptive adversary}: An adaptive adversary is limited to choosing the loss function $f_t$ of each round before observing the player’s choice $x_t$.
        \item \textit{Oblivious adversary}: An oblivious adversary is limited to choosing all the loss functions $f_1, f_2, \dots, f_T$ before the game starting.
    \end{enumerate}
    \item \textbf{Feedback}
    \begin{enumerate}
        \item \textit{Full information setting}: After suffering the loss $f_t(x_t)$, $f_t$ is revealed to the player as feedback.
        \item \textit{First-order setting}: After suffering the loss $f_t(x_t)$, a gradient oracle is revealed to the player as feedback, where the player can query the gradient of the loss function.
        \item \textit{Multi-query bandit setting}: After suffering the loss $f_t(x_t)$, a zeroth-order oracle is revealed to the player (i.e, $f_t$ is supplied to  the player as a black box), where the player can query the value of the loss function at more than 1 points as feedback.
        \item \textit{Single-query bandit setting}: After suffering the loss $f_t(x_t)$, only the loss value $f_t(x_t)$ is revealed to the player. This can be seen as $1$-query of a zeroth-order oracle.
    \end{enumerate}
    \item \textbf{The property of loss functions}: $\beta$-smooth, $\alpha$-strongly, exp-concave et al.
    \item \textbf{The property of the feasible set}.
\end{enumerate}
\end{framed}

\subsection{Related work}
For more information about online convex optimization,  one can refer to  Reference \cite{Hazan11}.
Also note that any algorithm for online convex optimization can be converted to an algorithm for stochastic convex optimization with similar guarantees \cite{Hazan2011C,Rakhlin2012,Shamir2013,Shamir2017}, and the quantum state learning problem can be benefited from online convex optimization \cite{aaronson2018,yang2020,chen2020}.
As shown in Table \ref{tab:RB}, below we review some work on online convex optimization that are closely related to the topic considered in this paper, which is organized according to the feedback fashion.
We put the first-order setting and the full information setting together because most of the work in full information setting only use the gradient information.

\noindent\textbf{First-order oracles/Full information.} In 2003, Zinkevich defined the online convex optimization model and showed that the online gradient descent could achieve $O((D^2+G^2)\sqrt{T})$ regret \cite{Zinkevich03}.
By modifying the original proof of \cite{Zinkevich03}, the regret can be improved to $O(DG\sqrt{T})$ by choosing a better learning rate if the diameter $D$ and the Lipschitz constant $G$ are known to the player.
Additionally, it has been known that the lower bound of this setting is $\Omega(DG\sqrt{T})$ (Theorem 3.2 of \cite{Hazan16}).
For strongly convex loss functions, Hazan et al. showed that $O(G^2\log{T})$ regret could be achieved by using online gradient descent \cite{Hazan2007}.
In the same paper, for exp-concave loss functions, Hazan et al. proposed an online Newton method which achieved $O(DGn\log{T})$ regret \cite{Hazan2007}.

\noindent\textbf{Single-query bandit setting. } In those work mentioned above, it is assumed that the player can get full information of loss functions as feedback, or has access to gradient oracles of loss functions.
Contrarily, online convex optimization in the single-query bandit setting was proposed by Flaxman et al. \cite{Flaxman2005}, where the only feedback was the value of the loss and  the adversary was assumed to be oblivious.
Note that in the bandit setting, a regret bound for any strategy against the completely adaptive adversary is necessarily $\Omega(T)$.
Thus, it needs to be assumed that the adversary is adaptive or oblivious, i.e. the adversary must choose the loss function before observing the player's choice or before the game starting, respectively.
The expected regret of Flaxman's algorithm is $O(\sqrt{DGCn}T^{3/4})$, where $C$ is the width of the range of loss functions.
In 2016, the dependence on $T$ was reduced to $O(DGn^{11} \sqrt{T}\log^4{T})$ by Bubeck and Eldan \cite{Bubeck2016} with the price that the dimension-dependence increased to $n^{11}$. In 2017, the dependence on $T$ and $n$ was balanced slightly to $O(DGn^{9.5} \sqrt{T}\log^{7.5} T )$ \cite{bubeck2017}. Recently, Lattimore proved that the upper bound was at most $O(DGn^{2.5} \sqrt{T}\log T )$ \cite{lattimore2020}, which improves the ones by \cite{Bubeck2016,bubeck2017}, but it still had a polynomial dependence on $n$.

\noindent\textbf{Multi-query bandit setting. } Better regret can be achieved if the player can query the value of the loss function at more than $1$ points in each round.
In 2010, Agarwal et al. \cite{Agarwal2010} considered the multi-query bandit setting and proposed an algorithm with an expected regret bound of $O((D^2+n^2G^2)\sqrt{T})$, where the player queries $O(1)$ points in each round.
In 2017, the upper bound was improved to $O(DG\sqrt{nT/k})$ by Shamir \cite{Shamir2017}, where the player queries $k$ points in each round. It is worth mentioning that the regret lower bound of zeroth-order online convex optimization with $O(1)$ queries each round is still not known very well, and the best upper bound still has additional dependence of dimension.

\begin{table}[]
    \centering
    \begin{tabular}{c|c|c|c}
        Paper & Feedback & Adversary & Regret \\
        \hline
        \cite{Zinkevich03,Hazan16} & Full information (first-order oracles) & Completely adaptive & $O(DG\sqrt{T})$ \\
        \cite{Hazan16} & Full information (first-order oracles) & Completely adaptive & $\Omega(DG\sqrt{T})$ \\
        \hline
        \cite{Flaxman2005} & Single query (zeroth-order oracles) & Oblivious & $O(\sqrt{DGCn}T^{3/4})$ \\
        \cite{Bubeck2016} & Single query (zeroth-order oracles) & Oblivious & $O(DGn^{11} \sqrt{T}\log T )$ \\
        \cite{bubeck2017} & Single query (zeroth-order oracles) & Oblivious & $O(DGn^{9.5} \sqrt{T}\log T )$ \\
        \cite{lattimore2020} & Single query (zeroth-order oracles) & Oblivious & $O(DGn^{2.5} \sqrt{T}\log T )$ \\
        \hline
        \cite{Agarwal2010} & $O(1)$ queries (zeroth-order oracles) & Adaptive & $O((D^2+n^2G^2)\sqrt{T})$ \\
        \cite{Shamir2017} & k-queries (zeroth-order oracles) & Adaptive  & $O(DG\sqrt{nT/k})$ \\
        This work & $O(n)$ queries (zeroth-order oracles) & Completely adaptive & $O(DG\sqrt{T})$ \\
        \hline
        This work & $O(1)$ queries (quantum zeroth-order oracles) & Completely adaptive & $O(DG\sqrt{T})$ \\
        \hline
        \multicolumn{4}{c}{$\uparrow$ General convex loss functions / $\alpha$-strongly convex loss functions $\downarrow$} \\
        \hline
        \cite{Hazan2007} & Full information (first-order oracles) & Completely adaptive & $O(G^2 \log T)$ \\
        \hline
        This work & $O(1)$ queries (quantum zeroth-order oracles) & Completely adaptive & $O(G^2 \log T)$
    \end{tabular}
    \caption{Regret bound for online convex optimization with different settings. For a strategy, the less feedback information that the player uses, the better; the stronger adversary the player faces, the better. The full information feedback model reveals the most information about the function, while the single-query feedback model reveals the least. The completely adaptive adversary is strongest, thus the corresponding models are the least restrictive ones in using, while the oblivious adversary is weakest, thus the corresponding models are the most stringent ones in using. Compared with the first block, the improvement of Algorithm \ref{QOCO} is with respect to the feedback. Compared with the second block, the improvement is with respect to the adversary and the dependence of $T$ and $n$, with the price of slightly stronger feedback. Compared with the third block, the improvement is with respect to the adversary and the dependence of $n$. The part above the row containing arrows is about general convex loss functions, while the part below is about strongly convex loss functions.}
    \label{tab:RB}
\end{table}

\subsection{Problem setting and our contributions}
\label{secCon}
In this paper, we present quantum algorithms for the online convex optimization problem (Subsection \ref{secOCO}) in the multi-query bandit setting, exploring  whether quantum advantages can be found. Here, an algorithm is allowed to query the zeroth-order oracle  multiple times after committing the prediction for getting feedback in each round.
A classical zeroth-order oracle $O_f$ to the loss function $f$, queried with a vector $x\in\mathcal{K}$,  outputs $O_f(x)=f(x)$.
A quantum zeroth-order oracle $Q_f$ is a unitary transformation that maps a quantum state $\ket{x}\ket{q}$ to the state $\ket{x}\ket{q+ f(x)}$, where $\ket{x}$, $\ket{q}$ and $\ket{q+ f(x)}$ are basis states corresponding to the floating-point representations of $x$, $q$ and $q+f(x)$.
Moreover, given the superposition input $\sum_{x,q}\alpha_{x,q}\ket{x}\ket{q}$,  by linearity the quantum oracle will output the state  $\sum_{x,q}\alpha_{x,q}\ket{x}\ket{q+ f(x)}$.

Unlike the previous work of the bandit setting, we do not need to limit the power of the adversary, namely, the adversary in our setting is \textit{completely adaptive}, who can choose $f_t$ after observing the player's choice $x_t$.
We assume that both the player and the adversary are quantum, which means that the adversary returns a quantum oracle as feedback and the player can use a quantum computer and query the oracle with a superposition input.
In addition, as usually in online convex optimization, we also make the following assumptions:
The loss functions are G-Lipschitz continuous, that is, $|f_t(x)-f_t(y)| \leqslant G \|y-x\|,\quad \forall x, y \in \mathcal{K}$; the feasible set $\mathcal{K}$ is bounded and its diameter has an upper bound $D$, that is, $\forall{x,y \in \mathcal{K}}, \|x-y\|_2 \leq D$. $\mathcal{K},D,G$ are known to the player.

Our main results are as follows.

\begin{itemize}
    \label{RQ}
    \item[i] In online convex optimization problems with quantum zeroth-order oracles, there exists a quantum randomized algorithm that can achieve the regret bound $O(DG\sqrt{T})$, by querying the oracle  $O(1)$ times in each round. (Theorem \ref{TQ})

    \label{RAQ}
    \item[ii] In online convex optimization problems with quantum zeroth-order oracles and $\alpha$-strongly convex loss functions, there exists a quantum randomized algorithm that can achieve the regret bound $O(G^2\log{T})$,
    by querying the oracle  $O(1)$ times in each round. (Theorem \ref{TAQ})
\end{itemize}

For completeness, we  also give a simple classical algorithm which guarantees $O(DG\sqrt{T})$ regret by consuming $O(n)$ queries in each round (see Theorem \ref{TC}).

From Table \ref{tab:RB} and the above results, one can see the following points:
\begin{itemize}
      \item Quantum algorithms outperform classical ones in the zeroth-order OCO model, since to our best knowledge the   optimal classical algorithm  with $O(1)$ queries in each round can only achieve $O(DG\sqrt{nT})$ regret \cite{Shamir2017}, where $n$ is the dimension of the feasible set, whereas our quantum algorithm has a better regret $O(DG\sqrt{T})$.
      \item  The quantum zeroth-order  oracle  is as powerful as the classical first-order oracle, since our quantum algorithm with only $O(1)$ queries to the zeroth-order quantum oracle in each round has achieved the  regret  lower bound $\Omega(DG\sqrt{T})$  of the  classical algorithms  with first-order oracles \cite{Hazan16}.

    \item Theorem \ref{TAQ} shows that for $\alpha$-strongly convex loss functions, the quantum algorithm with only $O(1)$ queries to the zeroth-order oracle in each round   can achieve the same regret $O(G^2\log{T})$ as the classical algorithms in the   first-order setting \cite{Hazan2007}.
\end{itemize}

The dependency relationship between lemmas and theorems are depicted in Figure \ref{Relation}.
The main idea is that  we first give a quantum algorithm  in Algorithm \ref{QOCO}, and then we show that the algorithm can guarantee the results mentioned above by choosing appropriate parameters.
The overall idea of choosing parameters is as follows: since we can't get the best fixed strategy $x^*$ directly when we analyze the regret bound, we prove a stronger property, namely the subgradient bound, to get the difference between the loss suffered by the player and the function value of the loss function at any point in the feasible set (the second column in Figure \ref{Relation}). Then we let the `any point' be $x^*$ and choose parameters to bound every term and make the regret as small as it can (the third column in Figure \ref{Relation}).

In technical aspect, we use Jordan's quantum gradient estimation method \cite{Jordan05} as the gradient estimator of each round. However, in the original 1-query version, the analysis was given by omitting the high-order terms of Taylor expansion of the function directly, which did not give any error bound we need.
Later, a version contained error analysis was given in \cite{Gilyen2019}, and was applied to the general convex optimization problem\cite{DeWolf18,Childs18}.
In those case, however, $O(\log{n})$ repetitions were needed to estimate the gradient/subgradient within a acceptable error.
It's obvious that it doesn't meet our requirement, namely $O(1)$-query.
For solving the OCO problem, firstly, we improve the analysis of the quantum gradient estimation method in Lemma \ref{QGB}, and show that $O(1)$ queries is enough in our problem, instead of $O(\log{n})$ repetitions.
This comes from the observation that, for each coordinate, at the expense of a weaker quality of approximation, the failure probability of a single repetition can be made small (To get the $\ell_1$ bound with high probability).
The worse approximation guarantee can then be fixed by choosing a finer grid.
Secondly, we introduce the uncompute part which recovers the ancillary registers to the initial states.
Since the quantum subroutine needs to be called many times in an online setting, it is necessary to recycle the quantum resource otherwise it will waste a substantial number of qubits.
Furthermore, in $\alpha$-strong case, we give a better subgradient bound by using strong convexity so that logarithmic regret can be guaranteed by the same algorithm with different parameters.

The rest of this paper is organized as follows. Section \ref{Q} is for the online convex optimization with quantum zeroth-order oracles; Section \ref{C} is for the online convex optimization with classical zeroth-order oracles. Notations and some extra definitions are listed in \ref{appDef} for the reader's benefit. Proofs, except those of our main theorems, are placed in \ref{app_proofs}.

\begin{figure}
    \centering
    \begin{tikzpicture}[node distance=20pt]
        \node[draw](L3){Lemma 3};
        \node[draw, below=of L3](L1){Lemma 1};
        \node[draw, below=of L1](L2){Lemma 2};
        \node[draw, right=of L3](L4){Lemma 4};
        \node[draw, right=of L4](T1){Theorem 1};
        \node[draw, right=of L1](L5){Lemma 5};
        \node[draw, right=of L5](T2){Theorem 2};
        \node[draw, below=of L2](L6){Lemma 6};
        \node[draw, right=of L6](L7){Lemma 7};
        \node[draw, right=of L7](T3){Theorem 3};

        \draw[->] (L3)--(L4);
        \draw[->] (L3)--(L5);
        \draw[->] (L4)--(T1);
        \draw[->] (L1)--(L4);
        \draw[->] (L1)--(L5);
        \draw[->] (L5)--(T2);
        \draw[->] (L2)--(L4);
        \draw[->] (L2)--(L7);
        \draw[->] (L6)--(L7);
        \draw[->] (L7)--(T3);

    \end{tikzpicture}
    \caption{The relation among the lemmas and theorems of this paper. For example, the three arrows before \boxed{\rm Lemma \ 4} indicate that Lemma 4 is derived from Lemma 1, 2 and 3. The lemmas in the first column are technical lemmas which analyze the evaluating error of the basic modules. The lemmas in the second column are middle lemmas which combine those before the arrows correspondingly and show the subgradient bound in each round of the algorithms for different settings. The theorems in the last column give the carefully chosen parameters and prove the regret bound for different settings.}
    \label{Relation}
\end{figure}
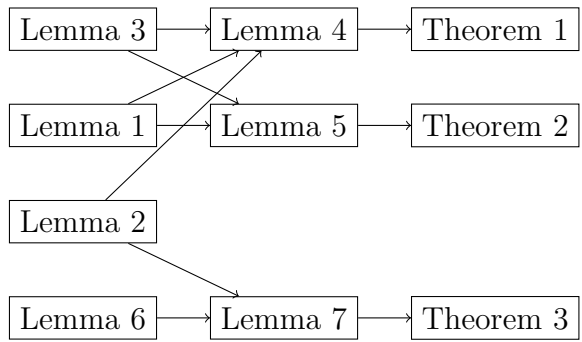

\section{Online convex optimization with quantum zeroth-order oracles}
\label{Q}
This section aims to prove Theorem \ref{TQ} and \ref{TAQ}. We first give a quantum algorithm and state some technical lemmas in Subsection \ref{secA}. Then in Section \ref{Q1}, for general convex loss functions, we show that by choosing appropriate parameters, sublinear regret can be guaranteed. Finally in Subsection \ref{Q3}, we show that for $\alpha$-strongly convex loss functions, $O(\log{T})$ regret can be guaranteed by Algorithm \ref{QOCO} with different parameters, which gives Theorem \ref{TAQ}. See  \ref{appDef} for the definition of $\alpha$-strongly convex functions.

\subsection{Algorithm}
\label{secA}
For the OCO problem stated in Subsection \ref{secOCO} and the setting stated in Subsection \ref{secCon}, given the total horizon $T$ and $\delta$,  we present Algorithm \ref{QOCO} to produce a decision sequence $x_1, x_2, x_3, \dots, x_T$ for the player, such  that it achieves a regret being sublinear of $T$, with probability greater than $1-\delta$.
Specifically, $\delta$ is divided into two parameters $p$ and $\rho$ which are two intermediate parameters used to adjust the success probability of two sub-processes (Lemma \ref{QGB} and \ref{Smooth}).
Initially, the algorithm chooses $x_1$ randomly from $\mathcal{K}$, and then sequentially produces  $x_2, x_3, \dots, x_T$ by   online gradient descent.  Steps 5-12 are the process of quantum gradient estimation.
The quantum circuit of $Q_{F_t}$ in Step 7 is constructed after the sampling of $z$ by using $Q_{f_t}$ twice;
$\mathbbm{1}$ in Step 7 is the $n$-dimensional all 1's vector;
the last register and the operation of addition modulo $2^c$ in Step 8 are used for implementing the common technique in quantum algorithm known as phase kickback which adds a phase shift related to the oracle;
Step 9 is known as uncompute trick which recovers the ancillary registers to the initial states so that they can be used directly in the next iterative;
the diagrammatic representation is depicted in Figure \ref{QGEC};
the projection operation is defined as $\hat{P}_{\mathcal{K}}(y) \triangleq  \mathop{\arg\min_{x \in \mathcal{K}}} \|x - y\|$; $B_{\infty}(x,r)$ is the ball in $L_{\infty}$ norm with radius $r$ and center $x$.

\begin{algorithm}[H]
    \caption{Quantum online subgradient descent (QOSGD)}
    \label{QOCO}
    \begin{algorithmic}[1]
        \REQUIRE Step sizes $\{\eta_t\}$, parameters $\{r_t\}, \{r^{\prime}_t\}$
        \ENSURE $x_1, x_2, x_3, \dots x_T$
      \STATE  Choose the initial point $x_1\in\mathcal{K}$ randomly.
        \FOR{$t=1$ to $T$}
            \STATE Play $x_t$, get the oracle of loss function $Q_{f_t}$ from the adversary.
            \STATE Sample $z\in B_{\infty}(x_t,r_t)$.
            \STATE Let $\beta=\cfrac{nG}{pr_t}$.  Prepare the initial state: $n$ $b$-qubit registers $\ket{0^{\otimes b},0^{\otimes b},\dots,0^{\otimes b}}$ where $b=\log_2 \cfrac{G\rho}{4\pi n^2 \beta r^{\prime}_t}$.
            Prepare $1$ $c$-qubit register $\ket{0^{\otimes c}}$ where $c=\log_2{\cfrac{4G}{2^b n \beta r^{\prime}_t }}-1$.
            And prepare $\ket{y_0}=\cfrac{1}{\sqrt{2^n}}\sum_{a\in\{0,1,\dots,2^n-1\}}e^{\cfrac{2\pi i a}{2^n}}\ket{a}$.
            \STATE Apply Hadamard transform to the first $n$ registers.
            \STATE Perform the quantum query oracle $Q_{F_t}$ to the first $n+1$ registers, where  $F_t(u)=\cfrac{2^b}{2Gr^{\prime}_t} \left [f_t \left (z+\cfrac{r^{\prime}_t}{2^b} \left (u-\cfrac{2^b}{2}\mathbbm{1} \right ) \right )-f_t(z) \right ]$, and the result is stored in the $(n+1)$th register.
            \STATE Perform the addition modulo $2^c$ operation to the last two registers.
            \STATE Apply the inverse evaluating oracle $Q_{F_t}^{-1}$ to the first $n+1$ registers.
            \STATE Perform quantum inverse Fourier transformations to the first $n$ registers separately.
            \STATE Measure the first $n$ registers in computation bases respectively to get $m_1,m_2,\dots,m_n$.
            \STATE Let $\widetilde{\nabla} f_t(x_t)=\cfrac{2G}{2^b} \left (m_1-\cfrac{2^b}{2},m_2-\cfrac{2^b}{2},\dots, m_n-\cfrac{2^b}{2} \right )^T$.
            \STATE Update $x_{t+1}=\hat{P}_{\mathcal{K}}(x_t-\eta_t \widetilde{\nabla} f_t(x_t))$.
            \STATE Bitwise erase the first $n$ registers with control-not gates controlled by the corresponding classical information of the measurement results $m_1,m_2,\dots,m_n$.
        \ENDFOR
    \end{algorithmic}
\end{algorithm}

First we analyze the query complexity of Algorithm \ref{QOCO}. In each round, it needs to call the oracle twice to construct $Q_F$, and twice to perform the uncompute step $Q_F^{-1}$, so totally $4$ times for computing the gradient. Thus, $O(1)$ times for each round.

In order to prove the main results of this paper (i.e., Theorems \ref{TQ} and \ref{TAQ}), some technical lemmas (i.e., Lemmas \ref{QGB}, \ref{SGB}, and \ref{Smooth} ) are required, of which the proofs  are presented in  \ref{app_proofs}.  First,  Lemma \ref{QGB} shows that the evaluating error of the gradient can be bounded. The proof sketch is: using the analysis framework of the phase estimation to get the error bound of each dimension of the ideal state; then using the trace distance to bound the difference in the probabilities between the ideal state and the current state of the algorithm with the help of smoothness; at last, using the union bound to get the one norm bound of the evaluating error.

\begin{figure*}[tbp]
    \includegraphics[width=16cm]{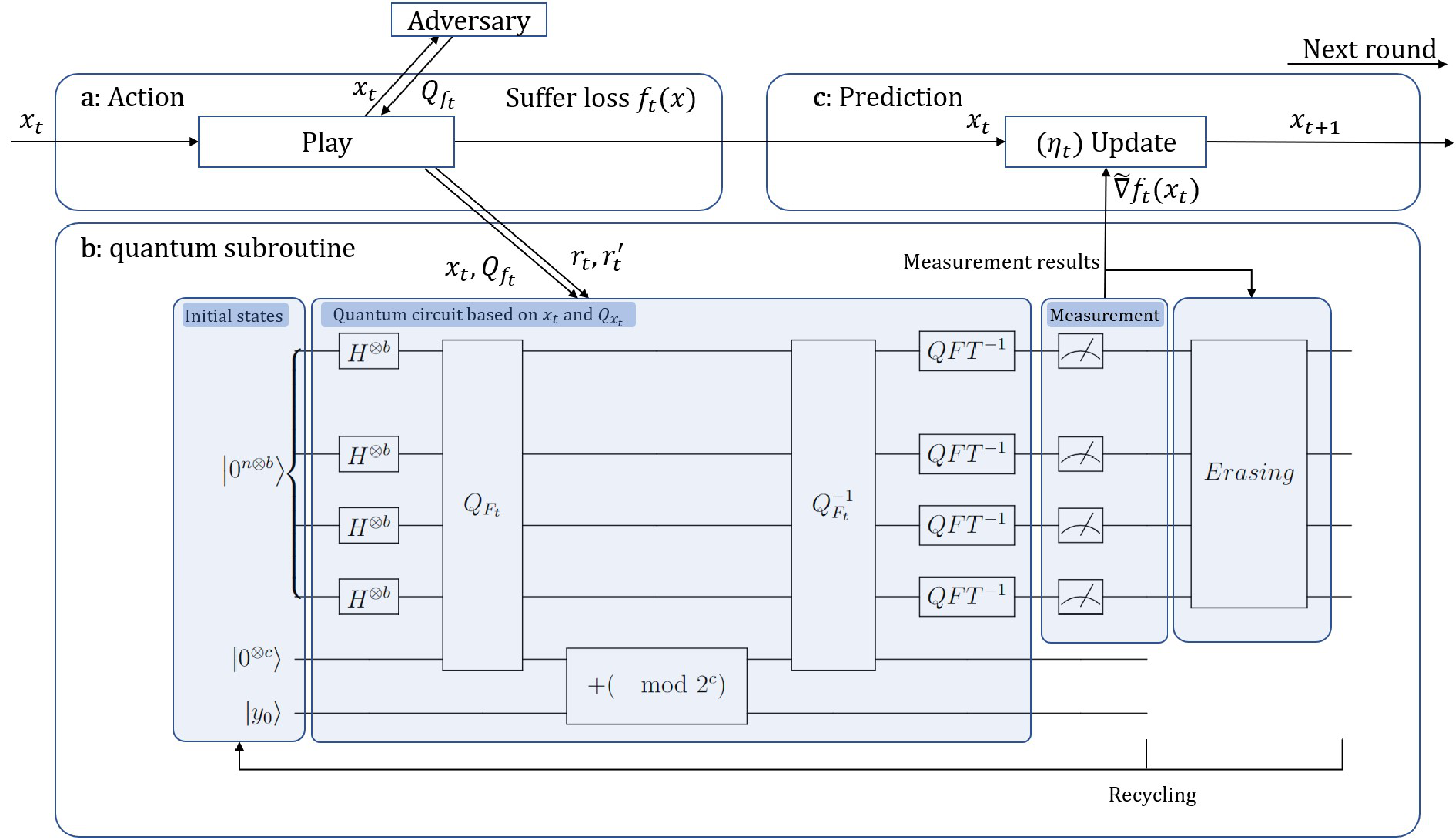}
\caption{Diagrammatic representation of Algorithm \ref{QOCO} in timestep $t$. The workflow can be divided into three parts: (a) Playing, feedback getting and loss suffering. This part is about the interaction between the player and the adversary. The decision and the corresponding feedback, with the parameters in this timestep, are then passed to the quantum subroutine. (b) The quantum subroutine for the gradient estimation. The initial quantum data structure is registers with all $\ket{0}$ in it, and the ancillary register with $\ket{y_0}$ in it is prepared by applying QFT to $\ket{1}$. The inputs (decision and feedback information) are encoded into the unitary transformation $Q_{F_t}$. Here the outputs on the right-hand side of the measurement operators stand for the quantum states after measurement, while the classical information of the measurement result is expressed by the outputs above the measurement part. The quantum register resources are recyclable through the uncompute step and the bitwise erasing, where the bitwise erasing is implemented with control-not gates controlled by the corresponding classical information of the measurement results. (c) Generating the decision for the next round. The prediction is generated by gradient descent with an appropriate learning rate. \label{QGEC}}
\end{figure*}

\begin{restatable}{mylem}{QGB}
\label{QGB}
    In Algorithm \ref{QOCO}, for all timestep $t$, if $f_t$ is $\beta$-smooth in the domain of $B_\infty(x_t,r_t+r^{\prime}_t)$,  then for any $r_t,r^{\prime}_t>0$ and $1\geq \rho_t >0$, the estimated gradient $\widetilde{\nabla} f_t(x_t)$ satisfies
    \begin{align}
        \Pr[\|{\nabla f_t(z)-\widetilde{\nabla} f_t(x_t)}\|_1>8 \pi n^3 (n/\rho_t +1) \beta r^{\prime}_t/\rho_t]<\rho_t.
    \end{align}

\end{restatable}

The evaluating error of subgradient of $f_t$ at point $x_t$ in each round can also be bounded by convexity and simple equivalent transformation, as follows.

\begin{restatable}{mylem}{SGB}
    \label{SGB}
    {\rm (Lemma 12 of \cite{DeWolf18})}
    In Algorithm \ref{QOCO}, for all timestep $t$, let $z\in B_{\infty}(x_t,r_t)$ and $f_t:\mathcal{K} \to \mathbb{R}$ be a convex function with Lipschitz parameter $G$,  then for any $y\in \mathcal{K}$,  the estimated gradient $\widetilde{\nabla} f_t(x_t)$ satisfies
    \begin{align}
        f_t(y)\geq f_t(x_t)+\widetilde{\nabla} f_t(x_t)^{\mathrm{T}}(y-x_t)-\|\nabla f_t(z)-\widetilde{\nabla} f_t(x_t)\|_1\|y-x_t\|_{\infty}  -2G\sqrt{n}r_t.
    \end{align}
\end{restatable}

Then, Lemma \ref{Smooth} shows that non $\beta$-smooth loss functions are still $\beta$-smooth in a small region with high probability by bounding the trace of their Hessian matrices \cite{Childs18,Hormander2015}.

\begin{restatable}{mylem}{Smooth}
\label{Smooth}
    {\rm (Lemma 2.5 and 2.6 of the old version of \cite{Childs18})}
    In Algorithm \ref{QOCO}, for all timestep $t$, let $f_t:\mathcal{K} \to \mathbb{R}$ be a convex function with Lipschitz parameter $G$. Then for any $r_t,r^{\prime}_t>0$ and $1 \geq p_t>0$, we have
    \begin{align}
        \Pr_{z\in B_{\infty}(x_t,r_t)} \left[\exists y\in B_{\infty}(z,r^{\prime}_t), \Tr{\nabla^2 f_t(y)} \geq \frac{nG}{p_t r_t}\right]
        \leq  p_t.
    \end{align}
\end{restatable}

\subsection{Analysis for general convex loss functions}
\label{Q1}

In this subsection we show how to choose  appropriate  parameters such that  Algorithm \ref{QOCO} guarantees $O(DG\sqrt{T})$ regret for all $T\geq 1$, which gives Theorem \ref{TQ}.  Before that, the following lemma is required, with proof given in  \ref{app_proofs}. This follows from combining Lemma \ref{QGB}, \ref{SGB}, and \ref{Smooth}.

\begin{restatable}{mylem}{QSGB}
    \label{QSGB}
    In Algorithm \ref{QOCO}, for all timestep $t$, let $f:\mathcal{K} \to \mathbb{R}$ is convex with Lipschitz parameter $G$, where $\mathcal{K}$ is a convex set with diameter $D$,  then for any $y\in \mathcal{K}$, with probability greater than $1-(\rho_t+p_t)$, the estimated gradient $\widetilde{\nabla} f_t(x_t)$ satisfies
    \begin{align}
        \label{QSGB1}
        f_t(y)\geq f_t(x_t)+\widetilde{\nabla} f_t(x_t)^{\mathrm{T}}(y-x_t)-\frac{8 \pi n^4 (n +\rho) DG r^{\prime}_t}{\rho_t^2 p_t r_t}-2G\sqrt{n}r_t.
    \end{align}
\end{restatable}

In the following, we prove the regret bound of Algorithm \ref{QOCO} for general convex loss functions.

\begin{mythe}
    \label{TQ}
    Algorithm \ref{QOCO} with parameters $\eta_t=\frac{D}{G\sqrt{t}}, r_t=\frac{1}{\sqrt{tn}}, r^{\prime}_t=\frac{\rho^2 p}{8 \pi \sqrt{tn^9} (n +\rho)}$, can achieve the regret bound $O(DG\sqrt{T})$,
    with probability greater than $1-T(\rho+p)$, and its query complexity is $O(1)$ in each round.
\end{mythe}

\begin{proof}
    Inequality (\ref{QSGB1}) is required to hold for all T rounds simultaneously. Let $B_t$ be the event that Algorithm \ref{QOCO} fails to satisfy Inequality (\ref{QSGB1}) in the $t$-th round. First, set the failure rate of each round to be the same, specifically, equal to $p+\rho$.
    Then by Lemma \ref{QSGB}, we have $\Pr(B_1)=\Pr(B_2)=\dots=\Pr(B_T)\leq p+\rho$.
    By the union bound (that is, for any finite or countable event set, the probability that at least one of the events happens is no greater than the sum of the probabilities of the events in the set), we have $\Pr(\cup_{t=1}^T B_t) \leq \sum_{t=1}^{T} \Pr(B_t) \leq T(p+\rho)$.
    Namely, the probability that Algorithm \ref{QOCO} fails to satisfy Inequality (\ref{QSGB1}) at least one round is  less than $T(\rho+p)$, which means the probability that Algorithm \ref{QOCO} succeeds for all $T$ round is greater than $1-T(\rho+p)$.
    Let $x^* \in \arg\min_{x\in \mathcal{K}} \sum_{t=1}^{T}f_t(x)$. Then by Lemma \ref{QSGB}, for the fixed $y=x^*$, with probability $1-T(\rho+p)$ we have
    \begin{align}
        \label{T4P1}
        f_t(x_t) - f_t(x^*)  \leq  \widetilde{\nabla} f_t(x_t)^{\mathrm{T}}(x_t-x^*)  + \frac{8 \pi n^4 (n +\rho) DG r^{\prime}_t}{\rho^2 p r_t}  + 2G\sqrt{n}r_t, \text{for all } t\in [T].
    \end{align}
    By the update rule for $x_{t+1}$ and the Pythagorean theorem, we get
    \begin{align}
        \label{T4P2}
        \|x_{t+1}-x^*\|^2 & = \|\hat{P}_{\mathcal{K}}(x_t-\eta_t\widetilde{\nabla} f_t(x_t))-x^*\|^2  \nonumber\\
        & \leq \|x_t-\eta_t \widetilde{\nabla} f_t(x_t) - x^*\|^2 \nonumber\\
        & = \|x_t-x^* \|^2 + \eta_t^2 \|\widetilde{\nabla} f_t(x_t)\|^2  -2\eta_t \widetilde{\nabla} f_t(x_t)^{\mathrm{T}}(x_t-x^*).
    \end{align}
    Hence
    \begin{align}
        \label{T4P3}
        \widetilde{\nabla} f_t(x_t)^{\mathrm{T}}(x_t-x^*)  \leq  \frac{\|x_t-x^* \|^2 - \|x_{t+1}-x^*\|^2}{2\eta_t} + \frac{\eta_t \|\widetilde{\nabla} f_t(x_t)\|^2}{2}.
    \end{align}
    Substituting Inequality (\ref{T4P3}) into Inequality (\ref{T4P1}) and summing Inequality (\ref{T4P1}) from $t=1$ to $T$, we have
    \begin{align}
        \label{T4P4}
        \sum_{t=1}^{T}(f_t(x_t)-f_t(x^*))
        & \leq \sum_{t=1}^{T}(\widetilde{\nabla} f_t(x_t)^{\mathrm{T}}(x_t-x^*) + \frac{8 \pi n^4 (n +\rho) DG r^{\prime}_t}{\rho^2 p r_t} + 2G\sqrt{n}r_t)  \nonumber \\
        & \leq \sum_{t=1}^{T}(\frac{\|x_t-x^* \|^2 - \|x_{t+1}-x^*\|^2}{2\eta_t} + \frac{\eta_t \|\widetilde{\nabla} f_t(x_t)\|^2}{2} \nonumber \\
        & \quad + \frac{8 \pi n^4 (n +\rho) DG r^{\prime}_t}{\rho^2 p r_t} + 2G\sqrt{n}r_t ).
    \end{align}
   Upper bounds can be obtained for the right side of the above inequality. First:
    \begin{align}
        \label{T2E5}
          \frac{1}{2} \sum_{t=1}^{T}\frac{\|x_t-x^* \|^2 - \|x_{t+1}-x^*\|^2}{\eta_t}
         & \leq  \frac{1}{2} \sum_{t=2}^{T} (\|x_t-x^* \|^2(\frac{1}{\eta_t}-\frac{1}{\eta_{t-1}})) +  \frac{\|x_t-x^* \|^2}{2\eta_1} \nonumber \\
         & \leq \frac{D^2}{2} \sum_{t=2}^{T}(\frac{1}{\eta_t}-\frac{1}{\eta_{t-1}}) +\frac{D^2}{2\eta_1} \nonumber \\
         & = \frac{D^2}{2\eta_T}.
    \end{align}
   Second, let $g=\nabla f_t(z)$, there is
    \begin{align}
        \label{T2E6}
          \sum_{t=1}^{T}\frac{\eta_t \|\widetilde{\nabla} f_t(x_t)\|^2}{2}
         & = \sum_{t=1}^{T}\frac{\eta_t \|\widetilde{\nabla} f_t(x_t)+g-g\|^2}{2} \nonumber\\
         & \leq \sum_{t=1}^{T}\frac{\eta_t (\|\widetilde{\nabla} f_t(x_t)-g\|+\|g\|)^2}{2} \nonumber\\
         & \leq \sum_{t=1}^{T}\frac{\eta_t (\|\widetilde{\nabla} f_t(x_t)-g\|_1+G)^2}{2} \nonumber\\
         & \leq \sum_{t=1}^{T}\frac{\eta_t (\frac{8 \pi n^4 (n +\rho) G r^{\prime}_t}{\rho^2 p r_t}+G)^2}{2},
    \end{align}
    where the last inequality holds as long as Lemma \ref{QSGB} holds (which implies that Lemma \ref{QGB} and \ref{Smooth} holds).

    Setting $\eta_t=\frac{D}{G\sqrt{t}}$, $r_t=\frac{1}{\sqrt{tn}}, r^{\prime}_t=\frac{\rho^2 p}{8 \pi \sqrt{tn^9} (n +\rho)}$, we have
    \begin{align}
          \sum_{t=1}^{T}(f_t(x_t)-f_t(x^*))
         & \leq \frac{DG\sqrt{T}}{2} + \sum_{t=1}^T\frac{D(G+G)^2}{2G\sqrt{t}} + \sum_{t=1}^T\frac{DG}{2 \sqrt{t}} + \sum_{t=1}^T \frac{2G}{\sqrt{t}} \nonumber\\
         & \leq \frac{DG\sqrt{T}}{2}+ \sum_{t=1}^T\frac{2DG}{\sqrt{t}}+ \frac{DG\sqrt{T}}{2} + 2G\sqrt{T} \nonumber\\
         & \leq \frac{DG\sqrt{T}}{2}+ 2DG\sqrt{T}+ \frac{DG\sqrt{T}}{2} + 2G\sqrt{T} \nonumber\\
         & = O(DG\sqrt{T}).
    \end{align}
    Thus, the theorem follows.
\end{proof}

 The space complexity can be analyzed as follows. Replacing $\{r_t\},\{r^{\prime}_t\}$ into $b,c$, we have $b=\log_2 \frac{G\rho}{4\pi n^2 \beta r^{\prime}_t}=\log_2 \frac{2n(n+\rho)}{\rho}=O(\log(Tn/\delta))$, $c=\log_2{\frac{4G}{2^b n \beta r^{\prime}_t }}-1=\log_2{\frac{16\pi n}{\rho }}-1=O(\log(Tn/\delta))$, where $\delta=T(\rho+p)$ is the failure probability we set for the algorithm.
 Since the failure probability of each timestep is set to be the same, the number of qubits we need in each timestep are actually equal.
 After the uncompute step and the bitwise erasing, the registers can be used for the next round directly, and no additional qubit is needed.
Thus, $O(n\log(Tn/\delta))$ qubits are needed totally.


\subsection{Analysis for $\alpha$--strongly convex loss functions}
\label{Q3}
In this subsection, we show that for $\alpha$-strongly convex loss functions,  $O(G^2\log{T})$ regret can be guaranteed by choosing different parameters in Algorithm \ref{QOCO}, which gives Theorem \ref{TAQ}.
Before that, the following lemma is required, with proof given in  \ref{app_proofs}. This follows from combining Lemma \ref{QGB}, \ref{Smooth}, and an improving version of Lemma \ref{SGB}.

\begin{restatable}{mylem}{alphaS}
    \label{alphaS}
    In Algorithm \ref{QOCO}, for all timestep $t$, let $f_t:\mathcal{K} \to \mathbb{R}$ be $\alpha$-strongly convex with Lipschitz parameter $G$, where $\mathcal{K}$ is a convex set,  then for any $y\in \mathcal{K}$, with probability greater than $1-(\rho_t+p_t)$, the estimated gradient $\widetilde{\nabla} f_t(x_t)$ satisfies
    \begin{align}
        \label{SGBAL}
        f_t(y)  \geq f_t(x_t)+\widetilde{\nabla} f_t(x_t)^{\mathrm{T}}(y-x_t)-\frac{8 \pi n^4 (n +\rho) DG r^{\prime}_t}{\rho_t^2 p_t r_t}  -(2G\sqrt{n}+\alpha nD)r_t + \frac{\alpha}{2}\|y-x_t\|^2.
    \end{align}
\end{restatable}

In the following, we show that with appropriate parameters, a better regret can be guaranteed for $\alpha$-strongly convex loss functions by Algorithm \ref{QOCO}.

\begin{mythe}
    \label{TAQ}
    For $\alpha$-strongly convex loss functions, Algorithm \ref{QOCO} with parameters $\eta_t=\frac{1}{\alpha t}, r_t=\frac{G^2}{t(2G\sqrt{n}+\alpha nD)}, r^{\prime}_t=\frac{G^2\rho^2 p}{8 \pi tn^4 (n +\rho)(2G\sqrt{n}+\alpha nD)}$, can achieve the regret bound $O(G^2\log{T})$,
    with probability greater than $1-T(\rho+p)$, and its query complexity is $O(1)$ in each round.
\end{mythe}

\begin{proof}
    Inequality (\ref{SGBAL}) is required to hold for all T rounds. Let $B_t$ be the event that Algorithm \ref{QOCO} fails to satisfy Inequality (\ref{SGBAL}) in the $t$-th round.
    First, set the failure rate of each round to be the same, specifically, equal to $p+\rho$.
    Then by Lemma \ref{alphaS}, we have $\Pr(B_1)=\Pr(B_2)=\dots=\Pr(B_T)\leq p+\rho$. By the union bound (that is, for any finite or countable event set, the probability that at least one of the events happens is no greater than the sum of the probabilities of the events in the set), we have $\Pr(\cup_{t=1}^T B_t) \leq \sum_{t=1}^{T} \Pr(B_t) \leq T(p+\rho)$. Namely, the probability that Algorithm \ref{QOCO} fails to satisfy Inequality (\ref{SGBAL}) at least one round is less than $T(\rho+p)$, which means the probability that Algorithm \ref{QOCO} succeeds for all $T$ round is greater than $1-T(\rho+p)$.
    Let $x^* \in \arg\min_{x\in \mathcal{K}} \sum_{t=1}^{T}f_t(x)$. By Lemma \ref{alphaS}, for the fixed $y=x^*$, with probability $1-T(\rho+p)$ we have
    \begin{align}
        \label{T5P1}
        f_t(x_t) - f_t(x^*)  \leq & \widetilde{\nabla} f_t(x_t)^{\mathrm{T}}(x_t-x^*) + \frac{8 \pi n^4 (n +\rho) DG r^{\prime}_t}{\rho^2 p r_t}  + (2G\sqrt{n}+\alpha nD)r_t \nonumber \\
        & - \frac{\alpha}{2}\|x_t-x^*\|^2.
    \end{align}
    By the update rule for $x_{t+1}$ and the Pythagorean theorem, there is
    \begin{align}
        \label{T5P2}
        \|x_{t+1}-x^*\|^2 & = \|\hat{P}_{\mathcal{K}}(x_t-\eta_t\widetilde{\nabla} f_t(x_t))-x^*\|^2  \nonumber\\
        & \leq \|x_t-\eta_t \widetilde{\nabla} f_t(x_t) - x^*\|^2 \nonumber\\
        & = \|x_t-x^* \|^2 + \eta_t^2 \|\widetilde{\nabla} f_t(x_t)\|^2  -2\eta_t \widetilde{\nabla} f_t(x_t)^{\mathrm{T}}(x_t-x^*).
    \end{align}
    Hence
    \begin{align}
        \label{T5P3}
        \widetilde{\nabla} f_t(x_t)^{\mathrm{T}}(x_t-x^*)  \leq & \frac{\|x_t-x^* \|^2 - \|x_{t+1}-x^*\|^2}{2\eta_t}  + \frac{\eta_t \|\widetilde{\nabla} f_t(x_t)\|^2}{2}.
    \end{align}
    Substituting Inequality (\ref{T5P3}) into Inequality (\ref{T5P1}) and summing Inequality (\ref{T5P1}) from $t=1$ to $T$, we have
    \begin{align}
        \label{T5P4}
        & \quad \sum_{t=1}^{T}(f_t(x_t)-f_t(x^*)) \nonumber \\
        & \leq \sum_{t=1}^{T}(\widetilde{\nabla} f_t(x_t)^{\mathrm{T}}(x_t-x^*) + \frac{8 \pi n^4 (n +\rho) DG r^{\prime}_t}{\rho^2 p r_t} + (2G\sqrt{n}+\alpha nD)r_t - \frac{\alpha}{2}\|x_t-x^*\|^2)  \nonumber \\
        & \leq \sum_{t=1}^{T}(\frac{\|x_t-x^* \|^2 - \|x_{t+1}-x^*\|^2}{2\eta_t} + \frac{\eta_t \|\widetilde{\nabla} f_t(x_t)\|^2}{2}  + \frac{8 \pi n^4 (n +\rho) DG r^{\prime}_t}{\rho^2 p r_t} \nonumber \\
        & \quad + (2G\sqrt{n}+\alpha nD)r_t - \frac{\alpha}{2}\|x_t-x^*\|^2).
    \end{align}
   In the right side of the above inequality, for the first term and the last term, we have
    \begin{align}
        \label{T5E5}
         &\sum_{t=1}^{T}(\frac{\|x_t-x^* \|^2 - \|x_{t+1}-x^*\|^2}{2\eta_t}- \frac{\alpha}{2}\|x_t-x^*\|^2)  \nonumber \\
         \leq & \frac{1}{2} \sum_{t=2}^{T} (\|x_t-x^* \|^2(\frac{1}{\eta_t}-\frac{1}{\eta_{t-1}}-\alpha)) + \|x_t-x^* \|^2 (\frac{1}{2\eta_1}-\frac{\alpha}{2}) \nonumber \\
         \leq & \frac{D^2}{2} \sum_{t=2}^{T}(\frac{1}{\eta_t}-\frac{1}{\eta_{t-1}}-\alpha) + D^2(\frac{1}{2\eta_1}-\frac{\alpha}{2})\nonumber \\
         = & \frac{D^2}{2}(\frac{1}{\eta_T}-\alpha T).
    \end{align}
    The handing of the second term is the same as Equation (\ref{T2E6}). Setting $\eta_t=\frac{1}{\alpha t}$, $r_t=\frac{G^2}{t(2G\sqrt{n}+\alpha nD)}, r^{\prime}_t=\frac{G^2\rho^2 p}{8 \pi tn^4 (n +\rho)(2G\sqrt{n}+\alpha nD)}$, we have
    \begin{align}
         \sum_{t=1}^{T}(f_t(x_t)-f_t(x^*))
        & \leq \frac{D^2}{2}(\alpha T-\alpha T) + \sum_{t=1}^T\frac{(G+G)^2}{2\alpha t} + \sum_{t=1}^T\frac{G^2}{2 t} + \sum_{t=1}^T \frac{G^2}{t} \nonumber\\
        & \leq \frac{2G^2}{\alpha}\log{T} + \frac{G^2\log{T}}{2} + G^2\log{T} \nonumber\\
        & = O(G^2\log{T}).
    \end{align}
   Hence, the theorem follows.
\end{proof}

\section{Online convex optimization with classical zeroth-order oracles}
\label{C}

In this section, we first give a classical OCO algorithm using classical zeroth-order oracles, which is stated in Algorithm \ref{COCO}. Then after some technical lemmas, we analyze its performance and show how we choose the appropriate parameters to ensure that it performs well in Theorem \ref{TC}.

Here we give the classical OCO algorithm. For the OCO problem stated in Subsection \ref{secOCO} and the setting stated in Subsection \ref{secCon}, given the total horizon $T$ and $\delta$, we present Algorithm \ref{COCO} to produce a decision sequence $x_1,x_2,x_3,\dots,x_T$ for the player, such that it achieves a regret being sublinear of $T$, with probability greater than $1-\delta$. Initially, the algorithm chooses $x_1$ randomly from $\mathcal{K}$, and then sequentially produces $x_2,x_2,x_3,\dots,x_T$ by online gradient descent.
Steps 5-8 are the process of the finite difference method.
Step 6 is the process of evaluating the partial derivative.
The projection operation in Step 8 is defined as $\hat{P}_{\mathcal{K}}(y) \triangleq  \mathop{\arg\min_{x \in \mathcal{K}}} \|x - y\|$; $B_{\infty}(x,r)$ is the ball in $L_{\infty}$ norm with radius $r$ and center $x$.

\begin{algorithm}[H]
    \caption{Classical online subgradient descent (COSGD)}
    \label{COCO}
    \begin{algorithmic}[1]
        \REQUIRE Step sizes $\{\eta_t\}$, parameters $\{r_t\}, \{r^{\prime}_t\}$
        \ENSURE $x_2, x_3, \dots x_T$
        \STATE Choose the initial point $x_1\in \mathcal{K}$ randomly.
        \FOR{$t=1$ to $T$}
            \STATE play $x_t$, get the oracle of loss function $O_{f_t}$.
            \STATE Sample $z\in B_{\infty}(x_t,r_t)$.
            \FOR{$j=1$ to $n$}
                \STATE $\nabla_j^{(r^{\prime}_t)}f_t(z)=\cfrac{O_{f_t}(z+r^{\prime}_te_j)-O_{f_t}(z-r^{\prime}_te_j)}{2r^{\prime}_t}$;
            \ENDFOR
            \STATE $\widetilde{\nabla} f_t(x_t)=\left(\nabla_1^{(r^{\prime}_t)}f_t(z),\nabla_2^{(r^{\prime}_t)}f_t(z), \dots,\nabla_n^{(r^{\prime}_t)}f_t(z)\right)$.
            \STATE update $x_{t+1}=\hat{P}_{\mathcal{K}}(x_t-\eta_t \widetilde{\nabla} f_t(x_t))$.
        \ENDFOR
    \end{algorithmic}
\end{algorithm}

First we analyze the query complexity of Algorithm \ref{COCO}. In each round, it needs to call the oracle twice to compute each partial derivative, so totally $2n$ times for computing the gradient. Thus, $O(n)$ times for each round.
Next, we show that Algorithm \ref{COCO} guarantees $O(DG\sqrt{T})$ regret for all $T\geq 1$ under the setting of our paper.

The evaluating error of gradient of $f_t$ at point $z$ in each round can be bounded as shown in Lemma \ref{FDMB}, of which the proof can be found in Lemma 10 and Lemma 11 of \cite{DeWolf18}.

\begin{restatable}{mylem}{FDMB}
    \label{FDMB}
    {\rm (Lemma 10 and 11 of \cite{DeWolf18})}
    In Algorithm \ref{COCO}, for all timestep $t$, let $f_t:\mathcal{K} \to \mathbb{R}$ be convex with Lipschitz parameter $G$. For any $r_t\geq r^{\prime}_t > 0$, the estimated gradient $\widetilde{\nabla} f_t(x_t)$ satisfies
    \begin{align}
        \mathbb{E}_{z\in B_{\infty}(x_t,r_t)}\|\nabla f_t(z)-\widetilde{\nabla} f_t(x_t)\|_1 \leq \frac{nGr^{\prime}_t}{2r_t}.
    \end{align}
\end{restatable}

The evaluating error of subgradient of $f_t$ at point $x_t$ in each round can also be bounded as follows, with proof given in  \ref{app_proofs}.

\begin{restatable}{mylem}{CSGB}
    \label{CSGB}
    In Algorithm \ref{COCO}, for all timestep $t$, let $f_t:\mathcal{K} \to \mathbb{R}$ be convex with Lipschitz parameter $G$, where $\mathcal{K}$ is a convex set with diameter $D$, then for any $y\in \mathcal{K}$, with probability greater than $1-\rho$, the estimated gradient $\widetilde{\nabla} f_t(x_t)$ satisfies
    \begin{align}
        \label{CSGB1}
        f_t(y)\geq f_t(x_t)+\widetilde{\nabla} f_t(x_t)^{\mathrm{T}}(y-x_t)-\frac{nGr_t^{\prime}D}{2\rho r_t}-2G\sqrt{n}r_t.
    \end{align}
\end{restatable}

Now we give the regret bound of Algorithm \ref{COCO}, with proof given in  \ref{app_proofs}.

\begin{restatable}{mythe}{TC}
    \label{TC}
    Algorithm \ref{COCO} with parameters $\eta_t=\frac{D}{G\sqrt{t}}, r_t=\frac{1}{\sqrt{tn}}, r^{\prime}_t=\frac{\delta}{T\sqrt{tn^3}}$ can achieve the regret bound $O(DG\sqrt{T})$,
    with probability greater than $1-\delta$, and its query complexity is $O(n)$ in each round.
\end{restatable}

\section{Conclusion}

In this paper, we considered the multi-points bandit feedback setting for online convex optimization against the completely adaptive adversary.
We provided a quantum algorithm and proved that it can achieve $O(\sqrt{T})$ regret where only $O(1)$ queries were needed in each round. We further showed that the algorithm can achieve $O(\log T)$ regret for $\alpha$-strongly convex loss functions by choosing different parameters.
These results showed that the quantum zeroth-order oracle is as powerful as the classical first-order oracle because the quantum case achieves the regret lower bound of the classical one, in the same query complexity.
Furthermore, our results showed that the quantum computing outperforms classical computing in this setting because, with $O(1)$ queries in each round, the state-of-art classical algorithm against weaker adversary only achieved $O(\sqrt{nT})$ regret.

This work leaves some open questions for future investigation: Can quantum algorithms achieve better regret bound if quantum first-order oracles are available? Is there any quantum algorithm which can achieve $O(\sqrt{T})$ regret with only $1$ query in each round?
Furthermore, the regret lower bound of classical $O(1)$ queries methods is still needed to be proved to show the quantum advantage rigorously.
It is also interesting to discuss some special cases of online convex optimization such as projection-free setting and constraint setting.

\section*{Acknowledgments}

This work was supported by the National Natural Science Foundation of China (Grant No. 61772565), the Basic and Applied Basic Research Foundation of Guangdong Province (Grant No. 2020B1515020050), the Key Research and Development Project of Guangdong Province (Grant No. 2018B030325001).

\appendix

\section{Notations and definitions}
\label{appDef}

In this appendix, we list some additional definitions for the reader's benefit. The definition or explanation of the notations used in this paper are listed in Table \ref{notation}.

\begin{table}[ht]
    \renewcommand\arraystretch{1.1}
    \begin{tabular}{c|p{12cm}}
        Notations & Definition or explanation \\
        \hline
        $\mathcal{K}$ & The convex feasible set. $\mathcal{K}\subseteq\mathbb{R}^n$. \\
        \hline
        $T$ & The total horizon. \\
        \hline
        $D$ & The upper bound of the diameter of the feasible set $\mathcal{K}$, that is,  $\forall{x,y \in \mathcal{K}}, \|x-y\|_2 \leq D$. \\
        \hline
        $G$ & The Lipschitz parameters of any possible loss function $f_t$, that is, $\forall x, y\in\mathcal{K}, |f_t(x)-f_t(y)| \leqslant G \|y-x\|$. \\
        \hline
        $\nabla f(x)$ & The gradient of $f$ at point $x$. \\
        \hline
        $\hat{P}_{\mathcal{K}}$ & The projection operation. $\hat{P}_{\mathcal{K}}(y) \triangleq  \mathop{\arg\min_{x \in \mathcal{K}}} \|x - y\|$. \\
        \hline
        $x_t$ & The decision made by the player in timestep $t$. \\
        \hline
        $f_t(x_t)$ & The loss suffered by the player in timestep $t$. \\
        \hline
        $Q_{f_t}$ & The feedback oracle got by the player from the adversary in timestep $t$. \\
        \hline
        $Q_{F_t}$ & The quantum circuit constructed by using $Q_{f_t}$ in timestep $t$. \\
        \hline
        $\widetilde{\nabla} f_t(x_t)$ & The gradient estimated by the quantum part in timestep $t$.\\
        \hline
        $\eta_t, r_t, r_t^{\prime}$ & The parameters to be chosen in timestep $t$.\\
        \hline
        $b,c$ & The number of qubits required in the registers, which are determined by the dimension $n$ and the success rate we set.\\
        \hline
        $p,\rho$ & The failure rates of the two sub-modules respectively. They can be controlled by the player and together with $T$ they will affect the final success rate. The higher the success rate it sets, the more qubits it needs.\\
    \end{tabular}
    \caption{Definition or explanation of notations used in this paper.}
    \label{notation}
\end{table}

\begin{mydef}[Convex set]
    A set $\mathcal{K}$ is convex if for any $x,y\in \mathcal{K}$,
    \begin{equation}
        \forall \theta\in[0,1], \theta x+(1-\theta)y\in\mathcal{K}.
    \end{equation}
\end{mydef}
\begin{mydef}[Convex functions]
    A function $f:\mathcal{K}\to \mathbb{R}$ is convex if for any $x,y\in\mathcal{K}$
    \begin{equation}
        \forall \theta\in[0,1], f((1-\theta)x+\theta y)\leq(1-\theta)f(x)+\theta f(y).
    \end{equation}
\end{mydef}
\begin{mydef}[$\beta$-smooth convexity]
    A function $f:\mathcal{K}\to \mathbb{R}$ is $\beta$-smooth if for any $x,y\in \mathcal{K}$,
    \begin{equation}
        f(y)\leq f(x)+\nabla f(x)^{\mathrm{T}}(y-x)+\frac{\beta}{2}\norm{y-x}^2.
    \end{equation}
\end{mydef}
\begin{mydef}[$\alpha$-strong convexity]
    A function $f:\mathcal{K}\to \mathbb{R}$ is $\alpha$-strongly convex if for any $x,y\in \mathcal{K}$,
    \begin{equation}
        f(y)\geq f(x)+\nabla f(x)^{\mathrm{T}}(y-x)+\frac{\alpha}{2}\norm{y-x}^2.
    \end{equation}
\end{mydef}
If the function is twice differentiable, the above conditions for $\beta$-smooth convexity and $\alpha$-strong convexity are equivalent to the following condition on the Hessian of $f$, denoted $\nabla^2 f(x)$:
\begin{equation}
    \alpha I \preceq \nabla^2 f(x) \preceq \beta I,
\end{equation}
where $A\preceq B$ if the matrix $B-A$ is positive semidefinite.
\begin{mydef}[Norm ball]
    The ball of radius $r>0$ in $L_p$ norm centered at $x\in \mathbb{R}^n$ is defined to be $B_p(x,r):=\{y\in \mathbb{R}^n| \norm{x-y}_p\leq r\}$.
\end{mydef}
\begin{mydef}[Subgradient]
    The subgradient for a function $f:\mathcal{K}\to \mathbb{R}$ at $x\in \mathcal{K}$ is defined to be any member of the set of vectors $\{\nabla f(x)\}$ that satisfies
    \begin{equation}
        \forall y\in \mathcal{K}, f(y)\geq f(x) + \nabla f(x)^{\mathrm{T}}(y-x).
    \end{equation}
\end{mydef}

\section{Proof of Lemmas and Theorems}
\label{app_proofs}

In this appendix, we give the proofs of lemmas and theorems mentioned in the text. Note that we omit the subscript $t$ in the proofs of the lemmas as they hold for each timestep $t$.

\QGB*

\begin{proof}
    The states after Step $5$ will be:
    \begin{align}
        \frac{1}{\sqrt{2^n}}\sum_{a\in\{0,1,\dots,2^n-1\}}e^{\frac{2\pi i a}{2^n}} \ket{0^{\otimes b},0^{\otimes b},\dots,0^{\otimes b}} \ket{0^{\otimes c}} \ket{a}.
    \end{align}
    After Step $6$:
    \begin{align}
       \frac{1}{\sqrt{2^{bn+c}}}\sum_{u_1,u_2,\dots,u_n\in\{0,1,\dots,2^b-1\}}
       \sum_{a\in\{0,1,\dots,2^c-1\}}
       e^{\frac{2\pi i a}{2^n}}
       \ket{u_1,u_2,\dots,u_n}
       \ket{0^{\otimes c}} \ket{a}.
    \end{align}
    After Step $7$:
    \begin{align}
       \frac{1}{\sqrt{2^{bn+c}}}\sum_{u_1,u_2,\dots,u_n\in\{0,1,\dots,2^b-1\}}
       \sum_{a\in\{0,1,\dots,2^c-1\}}
       e^{\frac{2\pi i a}{2^n}}
       \ket{u_1,u_2,\dots,u_n}
       \ket{F(u)} \ket{a}.
    \end{align}
    After Step $8$:
    \begin{align}
       \frac{1}{\sqrt{2^{bn+c}}}\sum_{u_1,u_2,\dots,u_n\in\{0,1,\dots,2^b-1\}}
       \sum_{a\in\{0,1,\dots,2^c-1\}}
       e^{2\pi i F(u)}
       e^{\frac{2\pi i a}{2^n}}
       \ket{u_1,u_2,\dots,u_n}
       \ket{F(u)} \ket{a}.
    \end{align}
    After Step $9$:
    \begin{align}
       \frac{1}{\sqrt{2^{bn+c}}}\sum_{u_1,u_2,\dots,u_n\in\{0,1,\dots,2^b-1\}}
       \sum_{a\in\{0,1,\dots,2^c-1\}}
       e^{2\pi i F(u)}
       e^{\frac{2\pi i a}{2^n}}
       \ket{u_1,u_2,\dots,u_n}
       \ket{0^{\otimes c}} \ket{a}.
    \end{align}
    In the following, the last two registers will be omitted:
    \begin{align}
       \frac{1}{\sqrt{2^{bn}}}\sum_{u_1,u_2,\dots,u_n\in\{0,1,\dots,2^b-1\}}
       e^{2\pi i F(u)}
       \ket{u_1,u_2,\dots,u_n} .
    \end{align}
    And then we simply relabel the state by changing $u \to v=u-\frac{2^b}{2}$:
    \begin{align}
        \label{phi}
        \frac{1}{\sqrt{2^{bn}}}\sum_{v_1,v_2,\dots,v_n\in\{-2^{b-1},-2^{b-1}+1,\dots,2^{b-1}\}}
        e^{2\pi i F(v)}
        \ket{v}.
    \end{align}
    We denote Formula (\ref{phi}) as $\ket{\phi}$. Let $g=\nabla f(z)$, and consider the idealized state
    \begin{align}
        \label{psi}
        \ket{\psi}=\frac{1}{\sqrt{2^{bn}}}\sum_{v_1,v_2,\dots,v_n\in\{-2^{b-1},-2^{b-1}+1,\dots,2^{b-1}\}}
        e^{\frac{2\pi i g \cdot v}{2G}}
        \ket{v}.
    \end{align}
    After Step $10$, from the analysis of phase estimation \cite{Brassard00}:
    \begin{align}
        \Pr[\abs{\frac{Ng_i}{2G}-m_i}>e]<\frac{1}{2(e-1)}, \forall{i\in [n]}.
    \end{align}
    Let $e=n/\rho +1$, where $1\geq \rho>0$. We have
    \begin{align}
        \Pr[\abs{\frac{Ng_i}{2G}-m_i}>n/\rho +1]<\frac{\rho}{2n}, \forall{i\in [n]}.
    \end{align}
    Note that the difference in the probabilities of measurement on $\ket{\phi}$ and $\ket{\psi}$ can be bounded by the trace distance between the two density matrices:
    \begin{align}
        \|\dyad{\phi}{\phi}-\dyad{\psi}{\psi}\|_1
        =2\sqrt{1-|\braket{\phi}{\psi}|^2}
        \leq 2\|\ket{\phi}-\ket{\psi}\|.
    \end{align}
    Since $f$ is $\beta$-smooth, we have
    \begin{align}
        F(v) & \leq \frac{2^b}{2Gr^{\prime}}[f(z+\frac{r^{\prime} v}{N})-f(z)]+\frac{1}{2^{c+1}} \nonumber\\
        & \leq \frac{2^b}{2Gr^{\prime}}[\frac{r^{\prime}}{2^b} g \cdot v +\frac{\beta (r^{\prime} v)^2}{2^{2b}}]
        +\frac{1}{2^{c+1}} \nonumber\\
        & \leq \frac{g \cdot v}{2G}+\frac{2^b \beta r^{\prime} n}{4G}+\frac{1}{2^{c+1}}.
    \end{align}
    Then,
    \begin{align}
        \|\ket{\phi}-\ket{\psi}\|^2 & = \frac{1}{2^{bn}} \sum_v | e^{2 \pi i F(v)} - e^{\frac{2 \pi i g \cdot v}{2G}} |^2 \nonumber \\
        & \leq \frac{1}{2^{bn}} \sum_v | 2 \pi i F(v) - \frac{2 \pi i g \cdot v}{2G} |^2 \nonumber \\
        & \leq \frac{1}{2^{bn}} \sum_v 4 \pi^2 (\frac{2^b \beta r^{\prime} n}{4G}+\frac{1}{2^{c+1}})^2.
    \end{align}
    Set $b=\log_2 \frac{G\rho}{4\pi n^2 \beta r^{\prime}}$, $c=\log_2{\frac{4G}{2^b n \beta r^{\prime} }}-1$.  We have
    \begin{align}
        \|\ket{\phi}-\ket{\psi}\|^2 \leq \frac{\rho^2}{16n^2},
    \end{align}
    which implies $\|\dyad{\phi}{\phi}-\dyad{\psi}{\psi}\|_1 \leq \frac{\rho}{2n}$. Therefore, by the union bound,
    \begin{align}
        \Pr[\abs{\frac{2^b g_i}{2G}-m_i}>n/\rho +1]<\frac{\rho}{n}, \forall{i\in [n]}.
    \end{align}
    Furthermore, there is
    \begin{align}
        \Pr[\abs{g_i-\widetilde{\nabla}_i f(x)}>\frac{2G(n/\rho +1)}{2^b}]<\frac{\rho}{n}, \forall{i\in [n]},
    \end{align}
    as $b=\log_2 \frac{G\rho}{4\pi n^2 \beta r^{\prime}}$, we have
    \begin{align}
        \Pr[\abs{g_i-\widetilde{\nabla}_i f(x)}>8 \pi n^2 (n/\rho +1) \beta r^{\prime}/\rho]<\frac{\rho}{n}, \forall{i\in [n]}.
    \end{align}
    By the union bound, we have
    \begin{align}
        \Pr[\|{g-\widetilde{\nabla} f(x)}\|_1>8 \pi n^3 (n/\rho +1) \beta r^{\prime}/\rho]<\rho,
    \end{align}
    which gives the lemma.
\end{proof}

\SGB*
\begin{proof}
    Let $g=\nabla f(z)$.  For any $y \in \mathbb{R}^n$, by convexity and simple equivalent transformation, we have
    \begin{align}
         f(y) & \geq f(z)+<g,y-z> \nonumber \\
            & =f(z)+<g,y-z>+(\widetilde{\nabla} f(x)^{\mathrm{T}}(y-x) -\widetilde{\nabla} f(x)^{\mathrm{T}}(y-x))+(f(x)-f(x)) \nonumber \\
            & =f(x)+\widetilde{\nabla} f(x)^{\mathrm{T}}(y-x)+(g-\widetilde{\nabla} f(x))^{\mathrm{T}}(y-x)  +(f(z)-f(x))+g^{\mathrm{T}}(x-z) \nonumber \\
            & \geq f(x)+\widetilde{\nabla} f(x)^{\mathrm{T}}(y-x)-\|g-\widetilde{\nabla} f(x)\|_1 \|y-x\|_{\infty}  -G\|z-x\|_2+\|g\|_2 \|x-z\|_2 \nonumber \\
            & \geq f(x)+\widetilde{\nabla} f(x)^{\mathrm{T}}(y-x)-\|g-\widetilde{\nabla} f(x)\|_1 \|y-x\|_{\infty}  -2G\sqrt{n}r.
    \end{align}
\end{proof}

\Smooth*
\begin{proof}
    This lemma comes from Lemma 2.5 and 2.6 of the old version of \cite{Childs18}. They improved their proof in the published version, but the old version is enough for us. We rewrite it here for reader's benefit.
    First, we use the mollification of $f$, an infinitely differentiable convex function with the same Lipschitz parameter as $f$, to approximate $f$ with the approximated error much less than the truncation error, by choosing appropriate width of mollifier.
    Then from Inequality (2.21) of \cite{Childs18}, we have
    \begin{align}
        \mathbb{E}_{z\in B_{\infty}(x,r)}\text{Tr}(\nabla^2f(z)) \leq \frac{nG}{r}.
    \end{align}
    By Markov's inequality, we have
    \begin{align}
        \Pr_{z\in B_{\infty}(x,r)}[\Tr(\nabla^2 f(z)) & \geq \frac{nG}{pr}] \leq p.
    \end{align}
    We denote the set $\{y|\text{Tr}(\nabla^2f(y)) \leq \frac{nG}{p r}\}$ as $Y$, and denote the measure of $Y$ as $\mathcal{M}(Y)$.
    Consider $z\in B_{\infty}(x,r)$, the probability that $z\in Y \subseteq B_{\infty}(x,r)$ is $\mathcal{M}(Y)/(2r)^n$.
    Define $B_{\infty}(Y,r^{\prime}):=\{z|\exists y\in B_{\infty}(z,r^{\prime}), y\in Y\}$, from the union bound, we have
    \begin{align}
        \mathcal{M}(B_{\infty}(Y,r^{\prime})) & \leq \mathcal{M}(Y) + \mathcal{M}(Y) \mathcal{M}(B_{\infty}(x,r^{\prime})) \nonumber \\
        & = (1+(2r^{\prime})^n)\mathcal{M}(Y).
    \end{align}
    Then,
    \begin{align}
        & \Pr_{z\in B_{\infty}(x,r)}[\exists y\in B_{\infty}(z,r^{\prime}), \Tr(\nabla^2 f(y)) \geq  \frac{nG}{pr}] \nonumber \\
         = & \frac{\mathcal{M}(B_{\infty}(x,r^{\prime}))}{(2r)^n} \nonumber \\
         \leq & (1+(2r^{\prime})^n)\frac{\mathcal{M}(Y)}{(2r)^n} \nonumber \\
         = & (1+(2r^{\prime})^n) \Pr_{z\in B_{\infty}(x,r)}[z \in Y] \nonumber \\
         = & (1+(2r^{\prime})^n) \Pr_{z\in B_{\infty}(x,r)}[\text{Tr}(\nabla^2f(z)) \geq \frac{nG}{p r}] \nonumber \\
         \leq & (1+(2r^{\prime})^n)p.
    \end{align}
    Since $r^{\prime} \ll 1$ (see Theorem \ref{TQ} and \ref{TAQ}), we omit that term approximately, which gives the lemma.
\end{proof}

\QSGB*
\begin{proof}
    By lemma \ref{QGB}, we have
    \begin{align}
        \|{g-\widetilde{\nabla} f(x)}\|_1\leq 8 \pi n^3 (n/\rho +1) \beta r^{\prime}/\rho,
    \end{align}
    succeeded with probability greater than $1-\rho$.

    By lemma \ref{Smooth}, for any $z\in B_{\infty}(x,r)$, the probability of $\forall y\in B_{\infty}(z,r^{\prime}), \nabla^2 f(y) < \frac{nG}{pr} I$ is greater than $1-p$. Thus, by the condition of $\beta$-smooth convex (see \ref{appDef}), we set $\beta=\frac{nG}{p r}$.
    By the union bound, the probability of both success is greater than $1-(\rho+p)$. Combining with Lemma \ref{SGB}, we have
    \begin{align}
        f(y)\geq & f(x)+\widetilde{\nabla} f(x)^{\mathrm{T}}(y-x) -\frac{8 \pi n^4 (n +\rho) DG r^{\prime}}{\rho^2 p r}-2G\sqrt{n}r,
    \end{align}
    succeeded with probability greater than $1-(\rho+p)$.
\end{proof}

\alphaS*
\begin{proof}
    Let $g=\nabla f(z)$. For any $y \in \mathbb{R}^n$ and $z\in B_{\infty}(x,r)$, by strong convexity,
    \begin{align}
        f(y) & \geq f(z)+<g,y-z> + \frac{\alpha}{2}\|y-z\|^2.
    \end{align}
    For the last term in the right side, we have
    \begin{align}
        \frac{\alpha}{2}\|y-z\|^2 & =  \frac{\alpha}{2}\|(y-x)-(z-x)\|^2 \nonumber \\
        & \geq  \frac{\alpha}{2}(\|y-x\|-\|z-x\|)^2 \nonumber \\
        & = \frac{\alpha}{2}(\|y-x\|^2+\|z-x\|^2-2\|y-x\|\|z-x\|) \nonumber \\
        & \geq \frac{\alpha}{2}(\|y-x\|^2-2\sqrt{n}\|y-x\|_{\infty}\|z-x\|) \nonumber \\
        & \geq \frac{\alpha}{2}(\|y-x\|^2-2\sqrt{n}D\sqrt{n}r)\nonumber \\
        & = \frac{\alpha}{2}(\|y-x\|^2-2nDr).
    \end{align}
    For other terms, by the same technique as Lemma \ref{SGB}, we have
    \begin{align}
         f(y) & \geq f(z)+<g,y-z> + \frac{\alpha}{2}\|y-z\|^2 \nonumber \\
            & \geq f(x)+\widetilde{\nabla} f(x)^{\mathrm{T}}(y-x)-\|g-\widetilde{\nabla} f(x)\|_1 \|y-x\|_{\infty}  -2G\sqrt{n}r + \frac{\alpha}{2}\|y-z\|^2 \nonumber \\
            & \geq f(x)+\widetilde{\nabla} f(x)^{\mathrm{T}}(y-x)-\|g-\widetilde{\nabla} f(x)\|_1 \|y-x\|_{\infty}  -2G\sqrt{n}r \nonumber \\
            & \quad + \frac{\alpha}{2}(\|y-x\|^2-2nDr) \nonumber \\
            & \geq f(x)+\widetilde{\nabla} f(x)^{\mathrm{T}}(y-x)-\|g-\widetilde{\nabla} f(x)\|_1 D  -(2G\sqrt{n}+\alpha nD)r + \frac{\alpha}{2}\|y-x\|^2.
    \end{align}
    By lemma \ref{QGB}, we have
    \begin{align}
        \|{g-\widetilde{\nabla} f(x)}\|_1\leq 8 \pi n^3 (n/\rho +1) \beta r^{\prime}/\rho,
    \end{align}
    succeeded with probability greater than $1-\rho$.
    By lemma \ref{Smooth}, we have $\beta=\frac{nG}{p r}$ succeeded with probability greater than $1-p$. Then by the union bound, the probability of both success is greater than $1-(\rho+p)$. which we have,
    \begin{align}
        f(y) & \geq f(x)+\widetilde{\nabla} f(x)^{\mathrm{T}}(y-x)-\|g-\widetilde{\nabla} f(x)\|_1 D  -(2G\sqrt{n}+\alpha nD)r + \frac{\alpha}{2}\|y-x\|^2 \nonumber \\
        & \geq f(x)+\widetilde{\nabla} f(x)^{\mathrm{T}}(y-x)-\frac{8 \pi n^4 (n +\rho) DG r^{\prime}}{\rho^2 p r}  -(2G\sqrt{n}+\alpha nD)r + \frac{\alpha}{2}\|y-x\|^2,
    \end{align}
    succeeded with probability greater than $1-(\rho+p)$, which gives the lemma.
\end{proof}

\CSGB*
\begin{proof}
    By Lemma \ref{FDMB} and Markov's inequality, we have
    \begin{align}
        \text{Pr}[\|\nabla f(z)-\widetilde{\nabla} f(x)\|_1 \leq \frac{nGr^{\prime}}{2r\rho}]\geq 1-\rho.
    \end{align}
    Combining with Lemma \ref{SGB}, we have
    \begin{align}
        f(y)\geq f(x)+\widetilde{\nabla} f(x)^{\mathrm{T}}(y-x)-\frac{nGr^{\prime}D}{2\rho r}-2G\sqrt{n}r.
    \end{align}
    succeeded with probability greater than $1-\rho$.
\end{proof}

\TC*
\begin{proof}
    Inequality (\ref{CSGB1}) is required to hold for all T rounds, then by the union bound, the probability that Algorithm \ref{COCO} fails to satisfy Inequality (\ref{CSGB1}) at least one round is less than $T\rho$, which means the probability that Algorithm \ref{COCO} succeeds for all $T$ round is greater than $1-T\rho$.
    Let $x^* \in \arg\min_{x\in \mathcal{K}} \sum_{t=1}^{T}f_t(x)$. By Lemma \ref{CSGB}, setting $\rho=\frac{\delta}{T}$, for the fixed $y=x^*$, with probability greater than $1-\delta$, we have
    \begin{equation}
        \label{T2P1}
        f_t(x_t) - f_t(x^*) \leq \widetilde{\nabla} f_t(x_t)^{\mathrm{T}}(x_t-x^*) + \frac{TnGr^{\prime}_tD}{2\delta r_t} + 2G\sqrt{n}r_t.
    \end{equation}
    By the update rule for $x_{t+1}$ and the Pythagorean theorem, we get
    \begin{align}
        \label{T2P2}
        \|x_{t+1}-x^*\|^2 & = \|\hat{P}_{\mathcal{K}}(x_t-\eta_t \widetilde{\nabla} f_t(x_t))-x^*\|^2  \nonumber \\
        & \leq \|x_t-\eta_t \widetilde{\nabla} f_t(x_t) - x^*\|^2 \nonumber \\
        & = \|x_t-x^* \|^2 + \eta_t^2 \|\widetilde{\nabla} f_t(x_t)\|^2 -2\eta_t \widetilde{\nabla} f_t(x_t)^{\mathrm{T}}(x_t-x^*).
    \end{align}
    Hence
    \begin{align}
        \label{T2P3}
        \widetilde{\nabla} f_t(x_t)^{\mathrm{T}}(x_t-x^*) \leq & \frac{\|x_t-x^* \|^2 - \|x_{t+1}-x^*\|^2}{2\eta_t}  + \frac{\eta_t \|\widetilde{\nabla} f_t(x_t)\|^2}{2}.
    \end{align}
    Substituting Inequality (\ref{T2P3}) into Inequality (\ref{T2P1}) and summing Inequality (\ref{T2P1}) from $t=1$ to $T$, we have
    \begin{align}
        \label{T2P4}
        \sum_{t=1}^{T}(f_t(x_t)-f_t(x^*))
        & \leq \sum_{t=1}^{T}(\widetilde{\nabla} f_t(x_t)^{\mathrm{T}}(x_t-x^*) + \frac{TnGr^{\prime}_tD}{2\delta r_t} + 2G\sqrt{n}r_t) \nonumber \\
        &\leq \sum_{t=1}^{T}(\frac{\|x_t-x^* \|^2 - \|x_{t+1}-x^*\|^2}{2\eta_t} + \frac{\eta_t \|\widetilde{\nabla} f_t(x_t)\|^2}{2}  + \frac{TnGr^{\prime}_tD}{2\delta r_t} \nonumber \\
        & \quad + 2G\sqrt{n}r_t ) .
    \end{align}
    Upper bounds can be obtained for the right side of the above inequality.
    The handing of the first term and the second term are the same as Inequality (\ref{T2E5}) (\ref{T2E6}).
    Setting $\eta_t=\frac{D}{G\sqrt{t}}$, $r_t=\frac{1}{\sqrt{tn}}, r^{\prime}_t=\frac{\delta}{T\sqrt{tn^3}}$, we have
    \begin{align}
          \sum_{t=1}^{T}(f_t(x_t)-f_t(x^*)) & \leq \frac{1}{2}DG\sqrt{T}+ \sum_{t=1}^T\frac{D(G+G)^2}{2G\sqrt{t}}+ \sum_{t=1}^T\frac{DG}{2 \sqrt{t}} + \sum_{t=1}^T \frac{2G}{\sqrt{t}} \nonumber\\
         & \leq \frac{1}{2}DG\sqrt{T}+ \sum_{t=1}^T\frac{2DG}{\sqrt{t}}+ \frac{DG\sqrt{T}}{2} + 2G\sqrt{T} \nonumber\\
         & \leq \frac{1}{2}DG\sqrt{T}+ 2DG\sqrt{T}+ \frac{DG\sqrt{T}}{2} + 2G\sqrt{T} \nonumber\\
         & = O(DG\sqrt{T}).
    \end{align}
    Hence, the theorem follows.
\end{proof}

\end{document}